\documentclass[a4paper, twocolumn]{article}
\usepackage[top=3cm,left=50pt,right=50pt,columnsep=2pc,heightrounded]{geometry}

\usepackage[utf8]{inputenc}
\usepackage{graphicx}
\usepackage{flushend}

\usepackage{amsmath}
\usepackage{amsthm}
\usepackage{amssymb}

\usepackage{cite}
\usepackage{hyperref}
\usepackage{algorithm}
\usepackage[noend]{algpseudocode}
\usepackage[nameinlink]{cleveref}

\newtheorem{theorem}{Theorem}[section]
\newtheorem{lemma}[theorem]{Lemma}
\newtheorem{corollary}[theorem]{Corollary}

\usepackage{xcolor}
\usepackage[normalem]{ulem}
\usepackage{enumitem}

\usepackage{caption}
\usepackage{subcaption} 
\usepackage[export]{adjustbox}
\usepackage{stfloats}

\hypersetup{hidelinks}

\providecommand{\comment}[1]{}
\renewcommand{\comment}[1]{{\color{red} #1}}
\newcommand{\littleModule}{\textit{A}}
\newcommand{\middleModule}{\textit{B}}
\newcommand{\bigModule}{\textit{C}}

\newcommand{\pattRandom}{\textit{`random'}}
\newcommand{\pattDec}{\textit{`decreasing'}}
\newcommand{\pattInc}{\textit{`increasing'}}

\newcommand{\onarXiv}[1]{}
\newcommand{\onaamas}[1]{}

\newtheorem{property}{Property}
\newtheorem*{property*}{Property}
\newtheorem{problem}{Problem}
\newtheorem{my_alg}{Algorithm}

\ifdefined\aamas
\theoremstyle{acmdefinition} 
\fi
\newtheorem*{definition*}{Definition}

\DeclareMathOperator{\opt}{OPT}
\DeclareMathOperator{\gen}{SOL}
\DeclareMathOperator{\genint}{SOL_\textit{int}}

\providecommand{\keywords}[1]{\textbf{\textit{Index terms---}} #1}

\renewcommand{\onarXiv}[1]{#1}

\begin{document}
	\title{\vspace{-1.5cm}Multirobot Coverage of Linear Modular Environments} 
	
	\author{
		Salaris, Mirko\\
		\small Politecnico di Milano\\
		\small Milan, Italy\\
		\small mirko.salaris@mail.polimi.it
		\and
		Riva, Alessandro\\
		\small Politecnico di Milano\\
		\small Milan, Italy\\
		\small alessandro.riva@polimi.it
		\and
		Amigoni, Francesco\\
		\small Politecnico di Milano\\
		\small Milan, Italy\\
		\small francesco.amigoni@polimi.it
	}
	
	\date{January 9, 2020\vspace{0.5cm}}
	\maketitle
	
	\begin{abstract}
		Multirobot systems for covering environments are increasingly used in applications like cleaning, industrial inspection, patrolling, and precision agriculture.
The problem of covering a given environment using multiple robots can be naturally formulated and studied as a multi-Traveling Salesperson Problem (mTSP).
In a mTSP, the environment is represented as a graph and the goal is to find tours (starting and ending at the same depot) for the robots in order to visit all the vertices with minimum global cost, which is typically calculated as the makespan, namely the length of the longest tour.
The mTSP is an NP-hard problem for which several approximation algorithms have been proposed. These algorithms usually assume generic environments, but tighter approximation bounds can be reached focusing on specific environments.
In this paper, we address the case of \emph{modular environments}, namely of environments composed of sub-parts, called modules, that can be reached from each other only through some linking structures.
Examples are multi-floor buildings, in which the modules are the floors and the linking structures are the staircases or the elevators, and floors of large hotels or hospitals, in which the modules are the rooms and the linking structures are the corridors.
We focus on \emph{linear} modular environments, with the modules organized sequentially, presenting an efficient (with polynomial worst-case time complexity) algorithm that finds a solution for the mTSP whose cost is within a bounded distance from the cost of the optimal solution.
The main idea of our algorithm is to allocate disjoint ``blocks'' of adjacent modules to the robots, in such a way that each module is covered by only one robot.
We experimentally compare our algorithm against some state-of-the-art algorithms for solving mTSPs in generic environments and show that it is able to provide solutions with lower makespan and spending a computing time several orders of magnitude shorter.
		\vfill
	\end{abstract}



\section{Introduction}




%

Several applications of autonomous multirobot systems require to perform some form of \emph{coverage}, namely to visit all the locations of given environments. Examples include cleaning~\cite{nikitenko-cleaning}, industrial inspection~\cite{correll-inspection}, patrolling~\cite{portugal-patrolling}, and precision agriculture~\cite{barrientos2011aerial}. The coverage problem has been widely studied in several variants~\cite{choset2001coverageSurvey,galceran2013survey}. One of its most common formulations, called \emph{multi-Traveling Salesperson Problem} (\emph{mTSP})~\cite{bektas2006multiple}, represents the environment with a graph and requires to find tours (starting and ending at a given vertex, called depot) such that, when the robots follow them, all the vertices are visited and the global cost, which is typically the makespan, namely the length of the longest tour, is minimized.
The mTSP is an NP-hard problem for which several approximation algorithms\footnotemark~have been proposed~\cite{Latah2016,nallusamy2010optimization,chandran2006Clustering,OKONJOADIGWE1988159}. The most known is arguably that by Frederickson \cite{frederickson1976approximation}, which provides an approximation factor of $\frac{5}{2} - \frac{1}{m}$, where $m$ is the number of robots. Such an approximation algorithm, similarly to several others, works in generic environments. In principle, tighter approximation bounds can be reached by adding constraints to the environment or focusing on specific classes of environments.
\footnotetext{An \emph{approximation algorithm} is an algorithm with polynomial-time complexity that finds approximate solutions to an NP-hard problem, providing theoretical guarantees on the bound of approximation.}

In this paper, we focus on \emph{modular environments}, namely on environments composed of sub-parts, called modules, that can be reached from each other only through some linking structures.
Examples of modular environments include multi-floor buildings, in which the modules are the floors and the linking structures are the staircases or the elevators, and floors of large hotels or hospitals, in which the modules are the rooms and the linking structures are the corridors.
In this paper, we consider ``linear'' instances of modular environments in which the modules are orderly aligned along a single linking structure, like multi-floor buildings with a single staircase.
We present an efficient (with polynomial worst-case time complexity) algorithm that finds a solution for the mTSP in modular environments whose cost is guaranteed to be at a bounded distance from the cost of the optimal solution. The bound depends on the shape of the modular environment. In particular, in environments in which covering the modules has a cost that is negligible wrt the cost of moving between modules, the approximation factor approaches $3/2$, that is the best known factor for the (single-robot) TSP \cite{christofides1976worst}. 
The main idea behind our algorithm is to allocate disjoint ``blocks'' of adjacent modules to the robots, in such a way that a robot visits and covers all the modules assigned to it (plus the portion of the linking structure needed to reach the modules) and that a module is covered only by a robot.
We experimentally compare our algorithm against some state-of-the-art algorithms for solving mTSPs in generic environments and show that our algorithm is able to provide solutions with lower makespan and spending a computing time several orders of magnitude shorter.

The original contributions of this paper are:
\begin{itemize}
\item the introduction of a new class of environments, called modular environments, that represent several relevant real-world environments (\Cref{sec:problem_formulation}),
\item the study of the mTSP in modular environments and, in particular, the analysis of \emph{integer solutions} that allocate disjoint ``blocks'' of adjacent modules to different robots (\Cref{subsec:int_sol}),
\item the definition of an efficient algorithm that calculates integer solutions for mTSPs in modular environments (\Cref{subsec:opt}),
\item the analysis of the approximation factor obtained by using integer solutions (\Cref{sec:approximation}),
\item the experimental assessment of the proposed algorithm, which shows that it outperforms state-of-the-art algorithms for mTSP in generic environments (\Cref{sec:experimental}).
\end{itemize}

\section{Related Work}
\label{sec:relatedwork}
To the best of our knowledge, the problem of covering modular environments or environments composed of repeated sub-structures has not been directly addressed in the literature. Here, we survey some works that have some relation to our problem.
We attempt to use a common terminology independent of the broad range of applications in which coverage tasks are encountered. In the following, the term `robot' will be used in place of agent, robot, or salesperson and the term `vertex' will be used to mean location, vertex, or city (this terminology is consistent with the rest of the paper, in which we will consider robots in an environment represented as a graph).

An overview of formulations and solutions for the mTSP is presented in~\cite{bektas2006multiple}. The basic definition of mTSP is the following: 
given a set of vertices and $m$ robots located at an initial vertex, called depot, the mTSP consists in finding tours for all the $m$ robots, which all start and end at the depot, such that all the vertices are visited at least once by any robot and the global cost of visiting all vertices is minimized. 
The cost metric can be defined as the total traveled distance or as the time required for completing all the tours. In the area of multirobot systems, we are usually interested in minimizing the total time of execution, namely the makespan.
A mTSP involves two main intertwined issues: how to partition the vertices among the robots and how to compute the optimal paths for the robots. The two main ways in which the mTSP is approached in the literature reflect this double nature of the problem and solve the two above issues in different orders.

An important theoretical result is that every mTSP can be approximately solved through a corresponding TSP formulation \cite{eilon1974distribution, svestka1973computational, bellmore1974transformation}, searching  for an optimal path for the TSP and then splitting it inducing a partition of the vertices on the robots.
The corresponding TSP formulation is obtained by creating $m$ copies of the original depot, each connected to the vertices adjacent to the original depot. The TSP solution obtained on this new graph is forced to have $m$ tours, namely a TSP path is built such that it visits each one of the $m$ copies of the original depot. If we ``cut'' the TSP path every time it visits a copy of the original depot, we obtain $m$ paths, each one starting from a copy of the depot and ending at a copy of the depot. These $m$ paths constitute the (non-optimal, in general) solution for the mTSP.

The other family of (generally non-optimal) approaches first group the vertices into $m$ clusters\footnotemark, so that each cluster represents a set of  vertices that are visited by a single robot whose path can be later optimized as in a TSP~\cite{Latah2016, nallusamy2010optimization, chandran2006Clustering, OKONJOADIGWE1988159}.
\footnotetext{\emph{Clustering} divides a set of objects into groups, or clusters, in such a way that objects belonging to the same cluster are similar to each other (according to some measure) and objects belonging to different clusters are dissimilar (according to the same measure). Proximity and distance measures can be used as similarity measures~\cite{chandran2006Clustering, anderberg1973cluster}.}
Clustering vertices has obvious computational advantages. The size of the search space for a routing problem over $n$ vertices is $\Omega(n!)$. Decomposing the problem into $k$ clusters, each one with approximately $n/k$ vertices, na\"{i}vely reduces the size of the search space to a function of $ k \times (n/k)! $, which is much smaller than $n!$ \cite{chandran2006Clustering}.

The mTSP can be also considered as a relaxation of the VRP (Vehicle Routing Problem), with the capacity constraints removed  \cite{bektas2006multiple}. 
The VRP asks for the optimal set of paths for a number of vehicles that have to deliver goods to a set of costumers, assuming that each vehicle has a limited capacity. The vehicles correspond to the robots of the mTSP and the costumers correspond to the vertices. Therefore, the equivalence between VRP and mTSP is obtained by setting the capacity of each vehicle infinitely large.
This implies that all the solving algorithms for the VRP are also valid for the mTSP, for example, see~\cite{yu2017minimum}.


Since the mTSP generalizes the TSP, it is NP-hard. Given that optimal solutions are likely to be out of reach for instances of realistic size, a great effort has been done for developing approximated and heuristic algorithms \cite{frederickson1976approximation,vandermeulen2019balancedTA,malik2007approximation,Latah2016}. For example, \cite{malik2007approximation} provides an algorithm with a constant approximation factor of $2$ for the Generalized, Multiple Depot, Multiple Traveling Salesman Problem (GMTSP), where the objective is to minimize the sum of the distances traveled by the robots.
With the same objective, \cite{Latah2016} approaches the mTSP through clustering and an Ant Colony Optimization algorithm.
Differently from the above works, in this paper, we are interested in finding solutions to a mTSP that minimize the makespan. For this variant of the mTSP, \cite{vandermeulen2019balancedTA} follows the clustering approach, with progressive improvement of the clusters and final optimization of the tours of the robots. Frederickson \cite{frederickson1976approximation} provided in 1979 a tour-splitting heuristic that yields an approximation factor of $\frac{5}{2} - \frac{1}{m}$ relying on the classical $\frac{3}{2}$-approximation to the TSP by Christofides \cite{christofides1976worst}. 
In the past 40 years, no better theoretical approximation factor has been found for the mTSP in which the objective function is the makespan. Thus research has proceeded restricting the problem to environments with constraints of practical interest. For instance, \cite{averbakh1997p} provides a $(2-2/(m+1))$-approximate algorithm for the makespan mTSP on trees with multiple depots. Our contribution follows a similar direction, considering modular environments.

\section{Problem formulation}
\label{sec:problem_formulation}

As already stated, mTSP is an NP-hard problem and cannot be solved efficiently under the assumption P$\neq$NP.
Therefore, we present an approximated algorithm that runs in polynomial time and provides tight approximation bounds for solving mTSPs in environments with a constrained structure, which the algorithm exploits.

We consider modular environments. A \emph{modular environment} is an environment constituted by sub-parts, the \emph{modules}, which can be repeated multiple times and which are connected to each other through some \emph{linking structures}.
The modules do not need to be all equal, but they need to be clearly identifiable and separable from each other. The idea is that moving from a module to another one could be done only through the linking structures and, thus, each module can be covered rather independently of the other modules. 

There are several examples of real-world environments that are modular in the above way.
For instance, a multi-floor building is a modular environment whose modules are the floors and the linking structures are the staircases (or the elevators).
On a smaller scale, consider floors of large hotels or hospitals. Each floor is itself a modular environment in which modules are the rooms and linking structures are the corridors. 
On a larger scale, modular environments can be identified in urban design and, specifically, in townhouses, where the modules are the houses and the linking structures are the streets.

In this paper, we consider modular environments whose linking structure is ``linear'', i.e., with the modules aligned in an ordered sequence along a single linking structure connecting one module to the next one. Examples of such linear modular environments are multi-floor buildings with a single staircase, floors of large hotels or hospitals whose rooms are connected by a single corridor, and townhouses that can be accessed by a single street.


\begin{figure}[t]
	\centering
	\includegraphics[width=0.5\columnwidth]{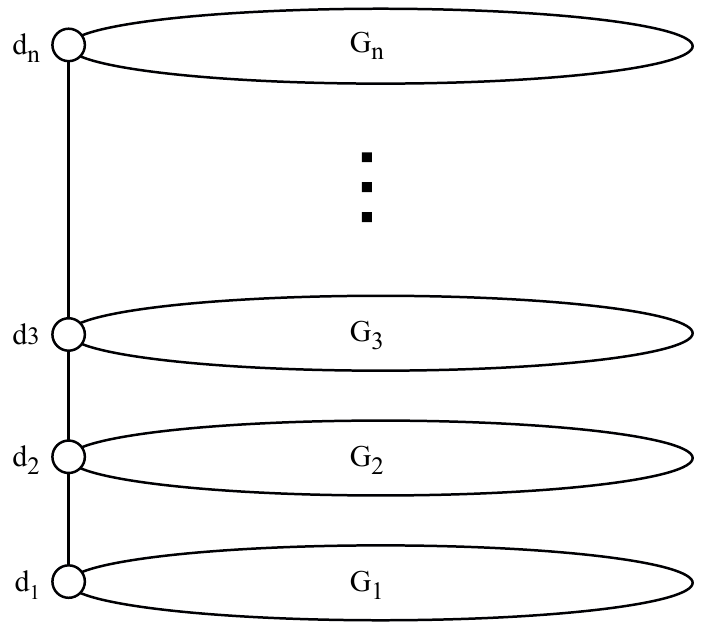}
	\caption{A schematic representation of the linear modular environments we consider. The linking structure (on the left) connects the doorways $d_{i}$, which are the entry points for modules $G_{1}, \ldots, G_{n}$. $d_{1}$ is also the depot.}
	\label{fig:modular_scheme}
	\vspace{-15pt}
\end{figure}

We call \emph{modular mTSP} a mTSP formulated on a linear modular environment.
The environment is thus composed of $n$ disjoint subgraphs $G_{i}=(V_{i},E_{i})$ (with $i=1, \ldots, n$), the modules.
In each module $G_{i}$, we identify a vertex $d_{i}$ that is the ``doorway'' to access the linking structure. Because of the linearity of the modular environment, the modules are orderly aligned and each module $G_{i}$ (except the first and the last ones) has exactly two adjacent modules, $G_{i-1}$ and $G_{i+1}$. The linking structure is represented by a set of $n-1$ edges $(d_{i},d_{i+1})$, with $i = 1, \ldots, n-1$.
An explanatory scheme is shown in \Cref{fig:modular_scheme}.

A metric $t$, representing the traveling time between locations, is defined over any pair of vertices in $V_i \times V_i$ and any pair $d_i, d_{i+1}$. 

We assume to have $m$ homogeneous robots. All robots start and eventually reach a \emph{depot}, which is assumed to be $d_{1}$. There is no constraint on the simultaneous presence of more robots at the same vertex or along the same edge. In applications, conflicts can be solved by using local collision-avoidance mechanisms.

Given all the above, we define our problem as the following optimization problem.


 
\begin{problem}[modular mTSP]\label{prb:general}
	Given a linear modular environment, assign to each robot a tour, starting from and ending to the depot, such that all the vertices of the modules are eventually covered with the minimum makespan.
\end{problem}


Let us here introduce an index to classify linear modular environments according to their shape:
\[\delta = \frac{\max_i t_\textit{tsp}(i)}{\sum_i t(d_i, d_{i+1})},\]
where $t_\textit{tsp}(i)$ is the time (calculated with the metric $t$) needed by a robot for entirely covering the $i$-th module $G_{i}$ when following an optimal tour.
For very ``wide'' instances, namely when the linking structure is short with respect to the size of the modules, $\delta \rightarrow \infty$.
Conversely, for very ``deep'' instances, in which the time for moving along the linking structure dominates the time for covering modules, $\delta \rightarrow 0$.

The results we present in this paper hold for all values of $\delta$. However, we note that, in a sense, instances with a deep structure are more interesting than instances with a wide structure, which can be easily turned into a generic mTSP.
\section{\onarXiv{\hspace{-3pt}}Integer Coverage of Modules}
\label{sec:integral_algorithm}

An intuitive way to solve the above problem is to partition the modules in ``blocks'' and assign them to the robots, in such a way that a robot visits and covers all the adjacent modules assigned to it (plus the portion of the linking structure needed to reach the modules) and that a module is covered only by a robot. We call such solutions \emph{integer solutions} or \emph{solutions in integer form} (formally defined below). Integer solutions are not guaranteed to be optimal as, in general, the optimum may require the robots to move back and forth among non-adjacent modules and cooperate for the coverage of each single module. However, we show that there is always an integer solution whose cost is guaranteed to be within a bound from the cost of the optimal solution.



\subsection{Integer Solutions}
\label{subsec:int_sol}

In the following, we will leverage solutions in integer form to develop approximated results to our modular mTSP.
As we will see, integer solutions turn out to be relatively simple to find and, under some conditions, surprisingly good\footnotemark.

\begin{definition*}[integer solution \emph{or} solution in integer form] \label{def:seqform}
	A solution is in integer form if for each robot $r$ there exist $i,j$ such that:
	\begin{itemize}
		\item for any $i \leq h \leq j$, the $h$-th module is entirely covered by $r$,
		\item $r$ does not take part to the coverage of any other module.
	\end{itemize}
\end{definition*}

\footnotetext{Notice that in integer solutions we can always consider $m \leq n$ even if we do not have this constraint in input, because the solution is trivial whenever $m > n$ . For any modular environment, there are no better integer solutions than the one with $m=n$, assigning one robot to each module and vice-versa.}
\begin{theorem}\label{thm:optint}
	Let $\opt$ be the makespan of an optimal solution for a given modular mTSP instance. Then, there must exist, for the same instance, an integer solution whose makespan $\genint$ satisfies:
	$$1 \leq \frac{\genint}{\opt} \leq 1 + \frac{\delta}{2}.$$
\end{theorem}
\begin{proof}
	The left-hand inequality clearly holds true, as any solution cannot be better than an optimal one.
	
	To prove the right-hand inequality, let us explicit the solutions as sequences of robot tours, i.e., $\opt = (s_1^*, \dots, s_m^*)$ and $\genint = (s_1, \dots, s_m)$.
	Let $\sigma(s)$ be the highest module index covered (partially or not) in a tour $s$. Without loss of generality, we assume that robots are ordered and, for any $r < r'$, it holds $\sigma(s_r) \leq \sigma(s_{r'})$ and $\sigma(s_r^*) \leq \sigma(s_{r'}^*)$.
	
	Let $T(s, i)$ be the time spent in tour $s$ for covering (partially or not) the $i$-th module.
	We construct $\genint$ from $\opt$ as follows.
	Define $s_1$ as the tour of robot $1$ that, starting from $d_1$, covers in a sequence all the modules (spending $t_{\textit{tsp}}(\cdot)$ for each module, i.e., the least time for a single-robot module coverage) until either the traveling time within modules reaches $\sum_{i} T(s^*_1, i)$ or module $\sigma(s_1^*)$ is covered.
	In order to preserve the integer form, the robot entirely covers the last-reached module (if any) before coming back to $d_1$. Accordingly, we define $s_r$, with $r > 1$, resuming the covering process from module $\sigma(s_{r-1})+1$.
	Leaving aside the time for moving between modules, the tour of a robot $r$ lasts $\sum_{i} T(s_r^*, i)$ plus the time needed to complete the coverage of the last module, or less if the module $\sigma(s_r^*)$ is reached (and covered) ahead of time.
	
We now give evidence that the so-obtained sequence of tours $\genint = (s_1, \dots, s_m)$ covers all the modules, i.e., $\sigma(s_m) = n$. Let $\bar{r} < m$ be the highest value such that $\sigma(s_{\bar{r}}) = \sigma(s_{\bar{r}}^*)$. By construction, in $\genint$ all the modules up to the $\sigma(s_{\bar{r}})$-th one have been covered and no robot $r > \bar{r}$ stops its covering ahead of time. Also, by construction of $\genint$ and by definition of $\sigma(\cdot)$, no robot $r > \bar{r}$ covers modules below the $\sigma(\bar{r})$-th one, and thus:
\begin{align*}
	\sum_{\substack{r > \bar{r} \\ i > \sigma(s_{\bar{r}})}} T(s_r, i)
	\geq
	\sum_{\substack{r > \bar{r} \\ \textit{any}~i}} T(s_r^*, i).
\end{align*}
Since the global time (i.e., the sum of all the robots' traveling times) needed to cover each module cannot be lower than $t_\textit{tsp}(i)$ for any $i$, we have:
\begin{align*}
	\sum_{\substack{r > \bar{r} \\ \textit{any}~i}} T(s_r^*, i)
	\geq
	\sum_{i > \sigma(s_{\bar{r}})} t_\textit{tsp}(i).
\end{align*}

Finally, given that in any integer solution robots cover \emph{entire distinct} modules, and in $\genint$ the coverage of each module takes $t_\textit{tsp}(i)$, all the modules $i > \sigma(s_{\bar{r}})$ must be covered as well.
	
In the solution outlined above, for any $r$ it holds $\sigma(s_r) \leq \sigma(s_r^*)$, that is, the time spent in $s_r$ for moving along the linking structure is not larger than the corresponding time spent in $s_r^*$.
	Furthermore, the time spent covering modules in any $s_r$ is not larger than $\sum_i T(s_r^*, i) + \max_{i} t_\textit{tsp}(i)$.
	Consequently:

\begin{align}
	\genint \leq \opt + \max_{1 \leq i \leq n} t_\textit{tsp}(i). \label{eq:optint1}
\end{align}
	
	Since any solution has to necessarily reach the last module and come back, we have a lower bound to the value of the optimum, namely $\opt \geq 2\sum_i t(d_i, d_{i+1})$.
	Making use of this last inequality in~(\ref{eq:optint1}), the claim of the theorem follows.
\end{proof}

Not only \Cref{thm:optint} highlights the existence of a relation between integer solutions and optimal ones, but it also quantifies this relation in terms of $\delta$, namely in terms of shape of the linear modular environment of the problem instance.
In particular, for environments that are deep rather than wide, the ratio between the cost of the optimal solution and that of its integer counterpart converges to $1$.

\subsection{Optimal Algorithm for Integer Solutions}
\label{subsec:opt}

In this section, we prove the following theorem.
\begin{theorem}\label{thm:main}
	If the optimal time $t_\textit{tsp}(i)$ for covering each module is given, then there exists an algorithm for finding a best integer solution in $\mathcal{O}(n^2 \log n \log m)$, i.e., in polynomial time with respect to the input problem size.
\end{theorem}
\noindent Before going through the proof we need some preliminary results. 
First of all, let us define $f(i, j, k)$ as the makespan of an optimal integer solution for $k$ robots that cover modules from $i$ to $j \geq i$, starting from $d_1$.
The computation of such a time is particularly simple in some special cases.
In particular, we highlight the following two cases.
\begin{property}\label{prp:case1}
	If only one module $i$ has to be covered, regardless the number of robots $k$ employed, it holds: $$f(i,i,k) = t_\textit{tsp}(i) + 2\sum_{h = 1}^{i-1} t(d_h, d_{h+1}).$$
\end{property}
\begin{property}\label{prp:case2}
	For any interval of modules to be covered $[i, j]$, if only one robot is employed, it holds: $$f(i,j,1) = \sum_{h = i}^j t_\textit{tsp}(h) + 2\sum_{h = 1}^{j-1} t(d_h, d_{h+1}).$$
\end{property} 
Computing generic values of function $f(i, j, k)$ requires some more effort.
Such values play a central role in the development of our polynomial algorithm, so we show a way to compute them quite quickly.

To this aim, we introduce the concept of \emph{split point} as the module at which a team of robots splits into two halves.
Intuitively, if, in the best integer solution, $k$ robots have to cover modules from $i$ to $j$, there must exist a value $i \leq h \leq j$ such that about half of the robots cover modules below the $h$-th one and about half of the robots cover modules above the $h$-th one.

\begin{definition*}[split point]
	Given $k$ robots and an (integer) interval $[i,j]$ of modules to cover, a split point $h \in \mathbb{N}$ is a solution of:
	$$
	f(i, j, k) = \min_{i \leq h \leq j} \max
	\begin{cases}
	f(i, h, \lfloor k/2 \rfloor)\\
	f(h+1, j, \lceil k/2 \rceil)
	\end{cases}.
	$$
\end{definition*}

Notice that, if for any $i \leq h \leq j$, the values $f(i, h, \lfloor k/2 \rfloor)$ and $f(h, j, \lceil k/2 \rceil)$ are known, a split point for $k$ robots covering modules from $i$ to $j$ can be found in $\mathcal{O}(j-i)$ by means of a linear inspection.

Once the split point is obtained, the value of $f(i, j, k)$ follows.
With this in mind, consider the algorithm sketched below.

\begin{my_alg}\label{alg:main}
Given a modular mTSP instance and the value of $t_\textit{tsp}(i)$ for each module $i$ of the instance:
	\begin{enumerate}[label=({\alph*})]
		\item Compute $f(i,j,k)$ for the cases $i = j$ and $k = 1$.\label{alg-in:one}
		\item Set $k = 2$ robots.\label{alg-in:two}
		\item For any $1 \leq i \leq j \leq n$ compute $f(i,j,k)$ and store the corresponding split points.\label{alg-in:three}
		\item Increment $k$ and repeat from \ref{alg-in:three} while $k \leq \lceil m/2 \rceil$.\label{alg-in:four}
		\item Compute the split point for $m$ robots visiting modules from $1$ to $n$.\label{alg-in:five}
    \item For each of the resulting halves, list recursively all the split points.\label{alg-in:six}
	\end{enumerate}
\end{my_alg}
A straightforward complexity of \Cref{alg:main} is $\mathcal{O}(mn^3)$.
In particular, the values of the function considered in~\ref{alg-in:one} can be filled as discussed at the beginning of this section in $\mathcal{O}(mn + n^2)$.
Each iteration of the loop~\ref{alg-in:two}-\ref{alg-in:four} has complexity linear in the number of robots and quadratic in the number of floors.
Also, since at a given iteration $k$ all the values of $f(i', j', k')$, with $k'<k$ and $i \leq i' \leq j' \leq j$, are known, the value of $f(i, j, k)$ and the corresponding split point can be computed in $\mathcal{O}(n)$.
The resulting complexity for~\ref{alg-in:two}-\ref{alg-in:four} is $\mathcal{O}(mn^3)$.
In~\ref{alg-in:five} the split point computation is $\mathcal{O}(n)$ as above. 
Finally, in~\ref{alg-in:six}, each recursive step takes a constant amount of time, since all the split points have been previously stored.
The number of recursive steps is $\mathcal{O}(2^{\log_2 m}) = \mathcal{O}(m)$ and the stated $\mathcal{O}(mn^3)$ complexity follows.

This computing time is slightly worse than that claimed by Theorem~\ref{thm:main}. However, we can start noticing there is no need for evaluating $f(i,j,k)$ in any $i,j,k$, as in the split point definition the number of robots is always divided by $2$. One can therefore pre-compute the set of needed values by means of a recursive procedure. The whole recursive width-expansion is bounded by $2$, as it holds:
\begin{align*}
	\left\lfloor \frac{\lceil k/2 \rceil}{2} \right\rfloor
	=
	\left\lceil \frac{\lfloor k/2 \rfloor}{2} \right\rceil.
\end{align*}
Since the depth of the recursion is bounded by $\mathcal{O}(\log m)$ (versus the $\mathcal{O}(m)$ complexity of the~\ref{alg-in:four} looping) the complexity of the algorithm can be lowered to $\mathcal{O}(n^3 \log m)$.

The computation of the values $f(i,j,k)$ and the corresponding split points can be sped up, as well. To this purpose, we need to point out that the two arguments of the $\max$ operator in the split point equation are, respectively, a monotonically non-decreasing and a monotonically non-increasing function of $h$. Consequently, their maximum value is minimized when the two functions are relatively close to each other. Before formalizing this concept, let us introduce a notation shortcut to ease the presentation:
\begin{align*}
l_{i,j}^k(h) &= f(i, h, \lfloor k/2 \rfloor),\\
u_{i,j}^k(h) &= f(h+1, j, \lceil k/2 \rceil).
\end{align*}
We can now formulate a sufficient condition that allows us to restrict the region of interest when searching for a split point.
\begin{lemma}\label{lmm:minimum}
	If $l_{i,j}^k(\hat{h}) \leq u_{i,j}^k(\hat{h})$ and $l_{i,j}^k(\hat{h}+1) \geq u_{i,j}^k(\hat{h}+1)$, then at least one between $\hat{h}$ and $\hat{h}+1$ is a split point.
\end{lemma}
\begin{proof}
	The statement follows immediately from the fact that $l_{i,j}^k(\cdot)$ ($u_{i,j}^k(\cdot)$) is a monotonically non-decreasing (non-increasing) function.
	Indeed, for any $h' \leq \hat{h}$ it holds $u_{i,j}^k(h') \geq u_{i,j}^k(\hat{h}) \geq l_{i,j}^k(\hat{h})$. Similarly, for any $h'' \geq \hat{h}+1$ we have $l_{i,j}^k(h'') \geq l_{i,j}^k(\hat{h}+1) \geq u_{i,j}^k(\hat{h}+1)$.
	Thus, in $\hat{h}$ and $\hat{h}+1$, the maximum between $l_{i,j}^k(\cdot)$ and $u_{i,j}^k(\cdot)$ is lower than or equal to their maximum computed everywhere else.
\end{proof}
This result allows us to reduce the search for an optimal split point to the search of a particular condition, namely, either when $l_{i,j}^k(\cdot)$ and $u_{i,j}^k(\cdot)$ are equal or when the latter exceeds the former.
Thanks to the monotonicity of the two functions, a binary search can be employed over the interval $[i,j]$ of integer values, improving the complexity of searching for a split point from $\mathcal{O}(j-i)$ to $\mathcal{O}(\log_2 (j-i))$.

\begin{proof}[Proof of \Cref{thm:main}]
	Since the intervals of modules that each robot has to cover can be easily computed while descending in the split point recursion of \Cref{alg:main}, an optimal integer solution to \Cref{prb:general} can be computed in $\mathcal{O}(n^2 \log n \log m)$.
\end{proof}

\section{Approximation}
\label{sec:approximation}

The results of the previous section hold when the optimal times $t_{\textit{tsp}}(i)$ for covering each module are known. If the time needed for covering a module is computed by means of a suboptimal algorithm, the values calculated in Properties \ref{prp:case1} and \ref{prp:case2} of \Cref{subsec:opt} could be much worse, harming the optimality of the integer solution calculated by \Cref{alg:main}.
However, if the TSP algorithm employed to calculate the tours covering each module has a bounded approximation factor, the approximation factor of an integer solution found by our approach is bounded as well.

\begin{theorem}
	If there exists an $\alpha$-approximation algorithm for the TSP, then there exists a polynomial-time algorithm for the modular mTSP, whose approximation factor is $\alpha \left( 1 + \frac{\delta}{2} \right)$.
\end{theorem}
\begin{proof}
	Let $\gen^{*}_\textit{int}$ and $\gen$ be the solutions found by \Cref{alg:main} leveraging, respectively, an optimal and an $\alpha$-approximation TSP algorithm when calculating the times for covering each module.
	Since, in the latter case, the computation of the split points makes use of suboptimal times, the resulting tours are up to $\alpha$ times worse than in the former case.
	Indeed, once the tour traveling times have been computed during step~\ref{alg-in:one} of \Cref{alg:main}, no other approximated information is introduced, and the final outcome is not worsened anymore.
	By making use of the relation $\gen \leq \alpha \gen^{*}_\textit{int}$ in \Cref{thm:optint}, we achieve the claimed approximation factor.
\end{proof}
From the Christofides' algorithm~\cite{christofides1976worst} for the TSP, we have the following direct consequence.
\begin{corollary}
	\label{T:bound-approx}
	There exists a $\frac{3}{2} \left( 1 + \frac{\delta}{2} \right)$-approximation algorithm for the modular mTSP.
\end{corollary}

It is interesting to compare this bound with that of Frederickson \cite{frederickson1976approximation}, which is $\frac{5}{2} - \frac{1}{m}$ and holds for mTSPs in any environment. For any instance of the modular mTSP in which $\delta < 1 - 1/m$, the bound of \Cref{T:bound-approx} is lower than that of Frederickson. In particular, when $\delta \rightarrow 0$, our bound approaches $\frac{3}{2}$ that represents the best-known approximation factor for the (single-robot) TSP.
\section{Experimental setup and results}

\begin{figure*}[t!]
	\captionsetup[subfigure]{labelformat=empty}
	\centering
	\begin{subfigure}[t]{0.3\textwidth}
		\centering
		\includegraphics[width=\columnwidth]{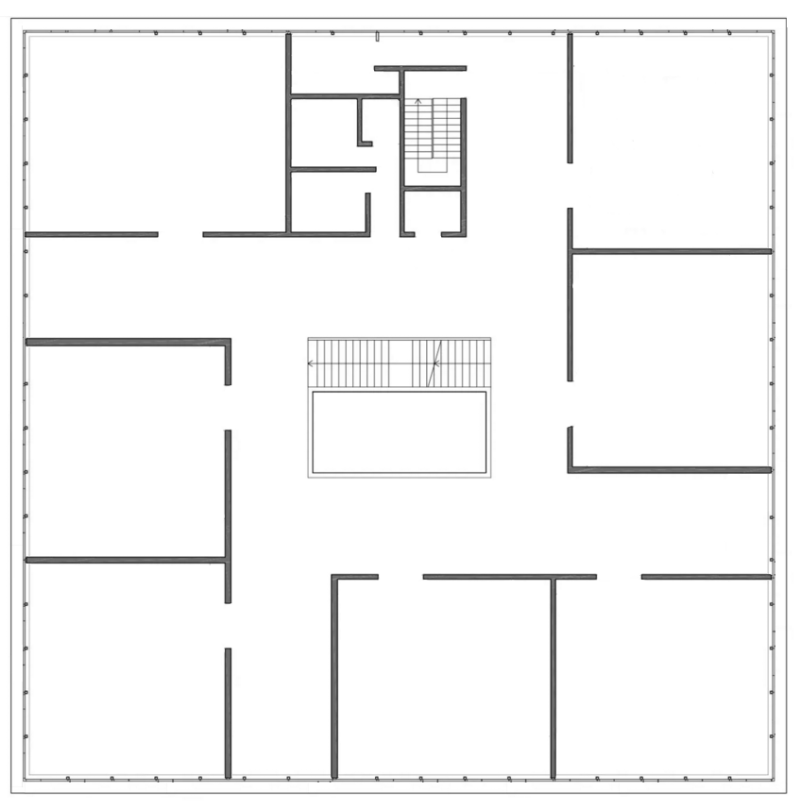}
		\caption{Module \littleModule}
		\vspace{10pt}
	\end{subfigure}
	\hfill
	\begin{subfigure}[t]{0.3\textwidth}
		\centering
		\includegraphics[width=\columnwidth]{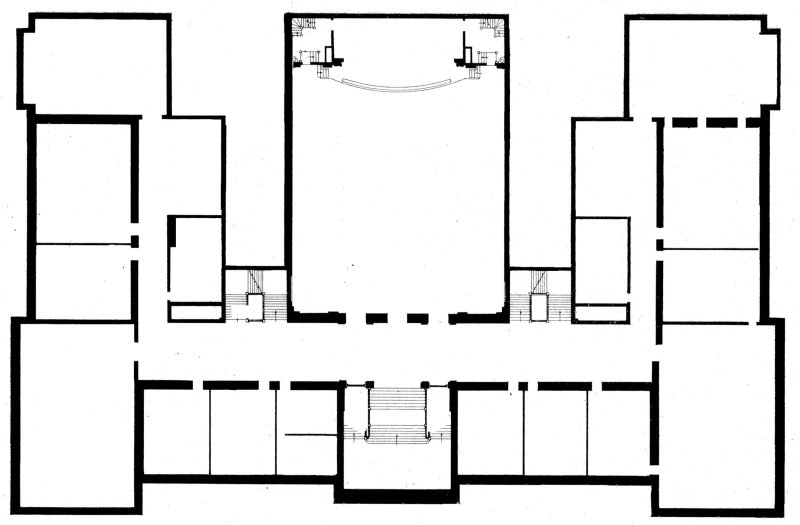}
		\caption{Module \middleModule}
		\vspace{10pt}
	\end{subfigure}
	\hfill
	\begin{subfigure}[t]{0.3\textwidth}
		\centering
		\includegraphics[width=\columnwidth]{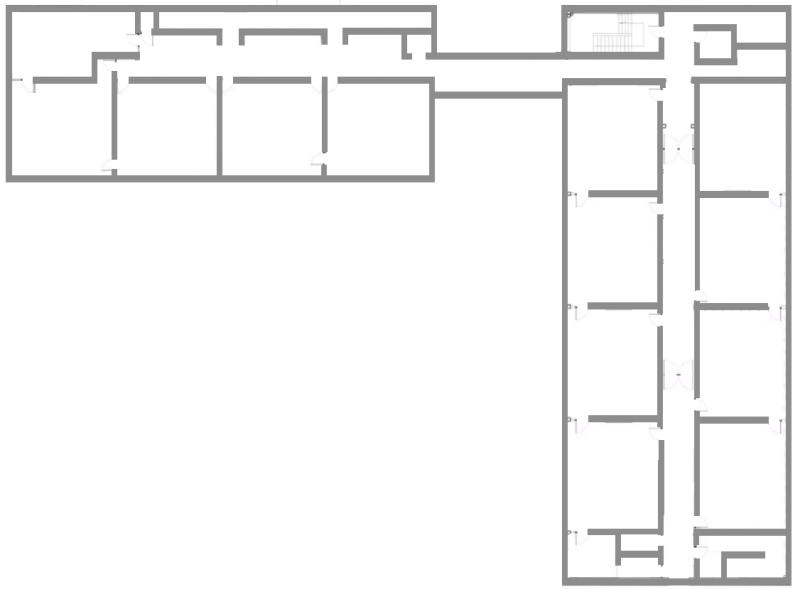}
		\caption{Module \bigModule}
	\end{subfigure}
	\vspace{-15pt}
	
	\caption{The floor plans of the three base modules used in experiments.
		Module \littleModule\ has a circular topology, the graph has $40$ vertices, and the approximated solution of the corresponding TSP is $198$ m long.
		Module \middleModule\ has a star topology, the graph has $47$ vertices, and the approximated solution of the corresponding TSP is $347$ m long.
		Module \bigModule\ has a linear topology, the graph has $80$ vertices, and the approximated solution of the corresponding TSP is $438$ m long.}
	\label{fig:modules}
\end{figure*}

\begin{figure}
	\centering
	\includegraphics[width=0.7\columnwidth]{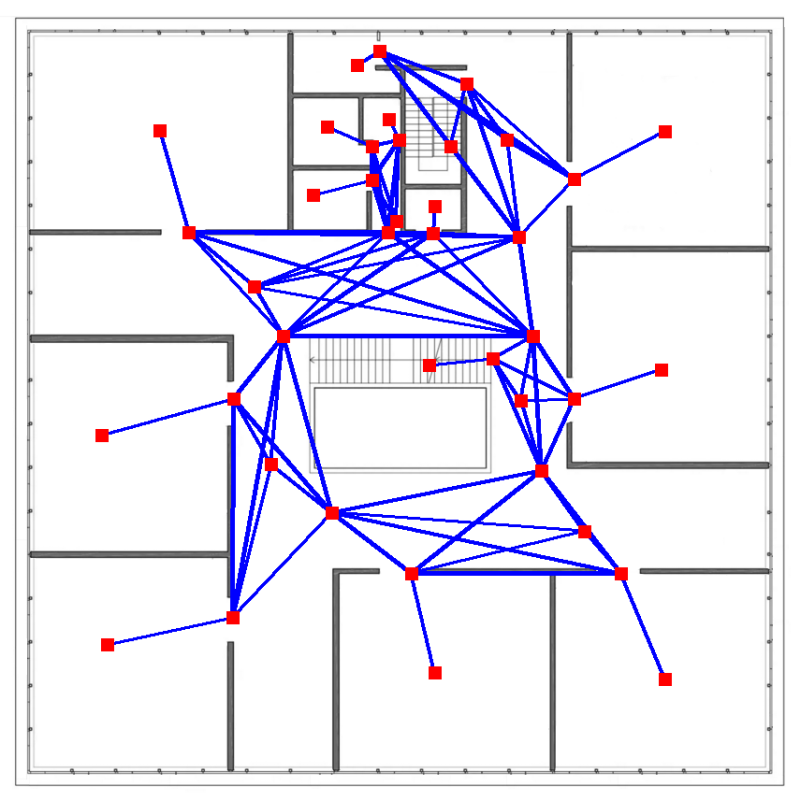}
	\caption{The floor plan of module \littleModule\ overlaid with the graph extracted from it. Vertices are in red and edges in blue.}
	\label{fig:graph_overlaid}
	\vspace{-10pt}
\end{figure}

\begin{figure*}[t]
	\begin{subfigure}[t]{0.3\textwidth}
		\centering
		\includegraphics[width=\columnwidth]{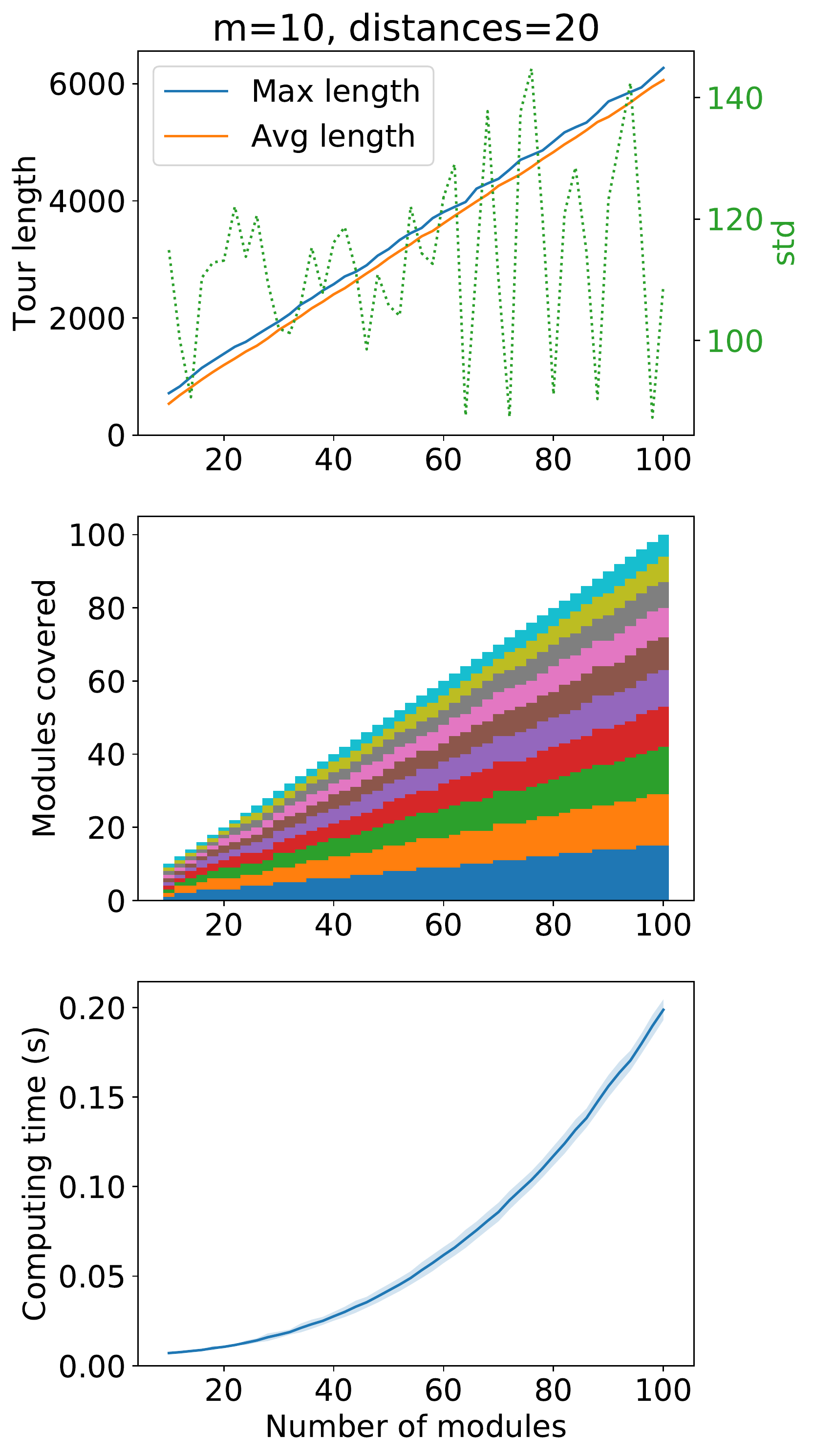}
		\caption{}
		\label{fig:hotel_varying_N}
		\vspace{10pt}
	\end{subfigure}
	\hfill
	\begin{subfigure}[t]{0.3\textwidth}
		\centering
		\includegraphics[width=\columnwidth]{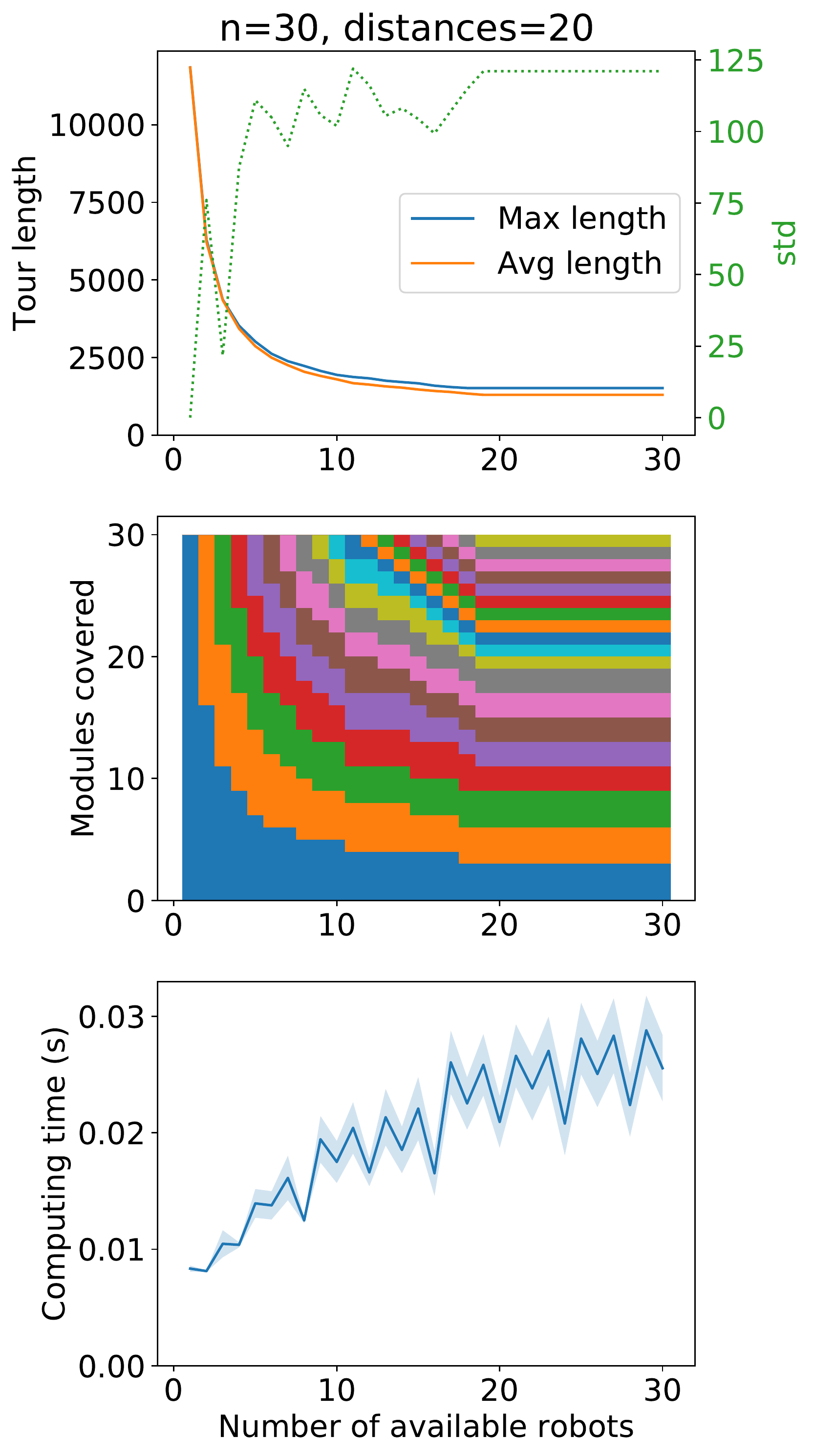}
		\caption{}
		\label{fig:hotel_varying_m}
		\vspace{10pt}
	\end{subfigure}
	\hfill
	\begin{subfigure}[t]{0.3\textwidth}
		\centering
		\includegraphics[width=\columnwidth]{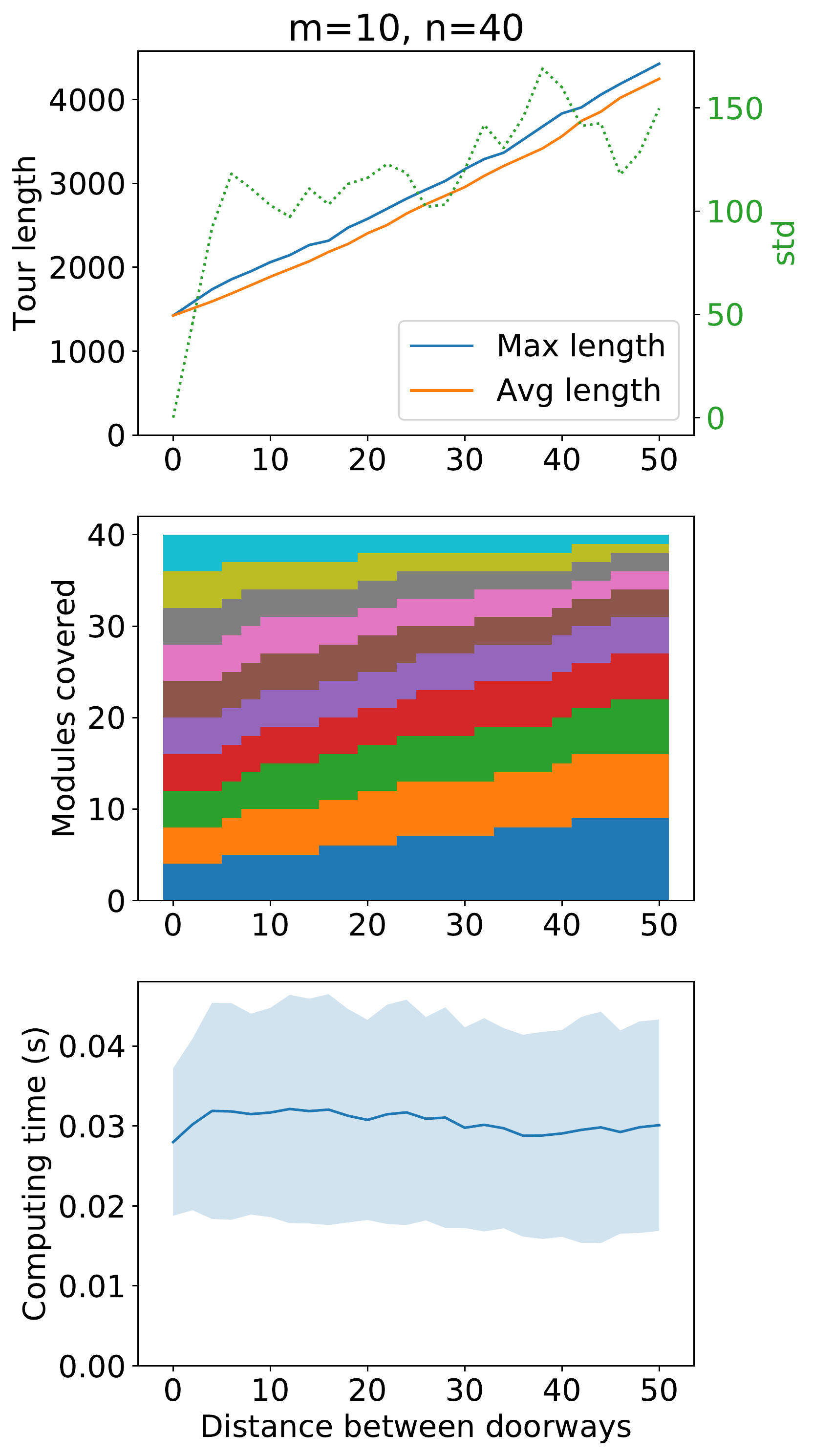}
		\caption{}
		\label{fig:hotel_varying_dist}
	\end{subfigure}
	\vspace{-15pt}
	\caption{
		Makespan (top row), robots-modules allocation (middle row), and computing time (bottom row) wrt varying the number of modules $n$ (\subref{fig:hotel_varying_N}), the number of robots $m$ (\subref{fig:hotel_varying_m}), and of the distance between doorways $t(d_i, d_{i+1})$ (\subref{fig:hotel_varying_dist}).
		The makespan plots show the makespan (max length), the average length of tours over the robots, and the corresponding standard variation.
		The robots-modules allocation plots have one color for each robot and show which modules are covered by which robots.
		The computing time plots show the average computing time, calculated over $50$ samples, and its standard deviation.}
	\label{fig:hotel_example}
\end{figure*}

\label{sec:experimental}
We compare our algorithm against two state-of-the-art mTSP algorithms: Frederickson \cite{frederickson1976approximation} and AHP-mTSP \cite{vandermeulen2019balancedTA}. These two algorithms are suboptimal and work for generic environments (see also \Cref{sec:relatedwork}). Frederickson works by computing a TSP over the whole environment and then splitting it into $m$ (the number of robots) tours.
AHP-mTSP starts from a random partition of the vertices of the graph of the environment in $m$ groups, each one assigned to a robot. It then applies a sequence of local operations (transfers, swaps, and improvements) on these groups, in order to balance the workload of the robots.   
To the best of our knowledge, as we discussed in \Cref{sec:relatedwork}, no algorithm specifically designed for modular environments is available in the literature.

Our algorithm and the Frederickson algorithm have been implemented in Python\footnote{Code available at \url{https://github.com/mirkosalaris/CoverageModularEnvironments}.}. For AHP-mTSP we use the original code, also written in Python and using the external Concorde solver \cite{concorde}, kindly provided by the authors of \cite{vandermeulen2019balancedTA}. For both Frederickson and our algorithm, we use the Christofides algorithm \cite{christofides1976worst} to compute approximated solutions to TSPs.
All computations have been performed on an AWS EC2 t2.large (2.3 GHz, 8 GB memory) instance with Ubuntu 16.04 AMI.

We consider environments built from a dataset of real environments, representing floors of schools\footnote{The dataset is obtained from a collection of blueprints and has been used in~\cite{Luperto2019}. The authors of~\cite{Luperto2019} kindly provided the dataset.}. Specifically, we build modular environments by repeating some base modules that represent three floors of real environments with different topologies and sizes (\Cref{fig:modules}). 
The graphs representing the base modules are built by manually extracting two types of vertices from the floor plans of the three real environments: centroids, that correspond to the geometrical centroids of individual rooms, and portals, that correspond to doors and passages between rooms. The centroid and the portals relative to the same room are connected by edges using Euclidean distance as metric $t$. Portals relative to the same room are connected to each other by edges using the $L_1$ norm as metric $t$. In this way, the graphs that represent the base modules reflect the structure and the shape of real environments, also in the case of non-convex rooms. The scale of the base modules, when not explicitly indicated, is estimated from the floor plans considering that doors are 80 cm wide. The metric $t$ is thus calculated assuming that robots move at constant speed (since it represents a time). 
 \Cref{fig:graph_overlaid} displays the graph extracted from module \littleModule\ overlapped to the floor plan of the real environment. For each module $i$, the doorway $d_{i}$ is selected randomly among the centroids of the module.

Before analyzing our algorithm and comparing it with the two state-of-the-art algorithms on complex modular environments, we consider a simple environment where all modules are identical.

\subsection{Environments with Identical Modules}
\label{subsec:experimental_single_repeated}

In this section, we consider modular environments composed of identical modules, like a tall building with identical floors. The repeated base module is module \middleModule, but results are similar for other base modules of \Cref{fig:modules}. 
We report the makespan of the solutions found by our algorithm, the number of modules allocated to each robot, and the computing time, according to the number of modules $n$, the number of robots $m$, and the distance between doorways of different modules $t(d_i, d_{i+1})$.
When not varying, we consider $n=30$ (or $n=40$), $m=10$, and $t(d_i, d_{i+1})=20$ for any $i$.
Results are shown in \Cref{fig:hotel_example}.

The computing times behave as expected wrt the number of modules $n$ (\Cref{fig:hotel_varying_N}) and the number of robots $m$ (\Cref{fig:hotel_varying_m}), and, as our algorithm abstracts from the values of the distances, they are rather independent of the distances between doorways (\Cref{fig:hotel_varying_dist}).
From the makespan plots, the standard deviation of the tour lengths over the robots is small for all the cases (notice that the scale of the standard deviation is much smaller than the scale of the tour length), showing that our algorithm performs a balanced division of effort over all the robots. This is confirmed by the curves showing the makespan (max length) and the average length of tours, which are very close to each other.

\Cref{fig:hotel_varying_m} shows that the advantage of using one further robot decreases when the total number of robots is already large. In this specific example with $n=30$, when $m > 18$, only $18$ robots are actually used because using more robots would not yield a better makespan.
Indeed, in any environment, the minimum makespan achievable is given by the time it takes for one robot to reach the last module, cover it, and return to the depot.
Whenever this makespan is reached, there are no advantages in using more robots.

From \Cref{fig:hotel_varying_dist}, we observe that when the distance between doorways is set to $0$, the standard deviation of the tour lengths is zero.
This is due to the fact that the number of modules is a multiple of the number of robots and that, in this case, all the modules are virtually connected to the depot $d_{1}$ because their doorways can be reached at no cost.
In the robots-modules allocation plot, we notice that, as the distance between doorways increases, modules near to the depot tend to be allocated in big chunks to the same robot and modules far from the depot tend to be allocated individually to different robots.
Indeed, when increasing the effort to reach the last module from the depot, the minimum achievable makespan increases.
This implies that robots that cover modules close to the depot have more time to cover multiple modules.


\subsection{Complex Modular Environments}
\label{subsec:experimental_general}
\begin{figure*}[t]
	\centering
	\begin{subfigure}[t]{0.3\textwidth}
		\centering
		\includegraphics[width=\textwidth]{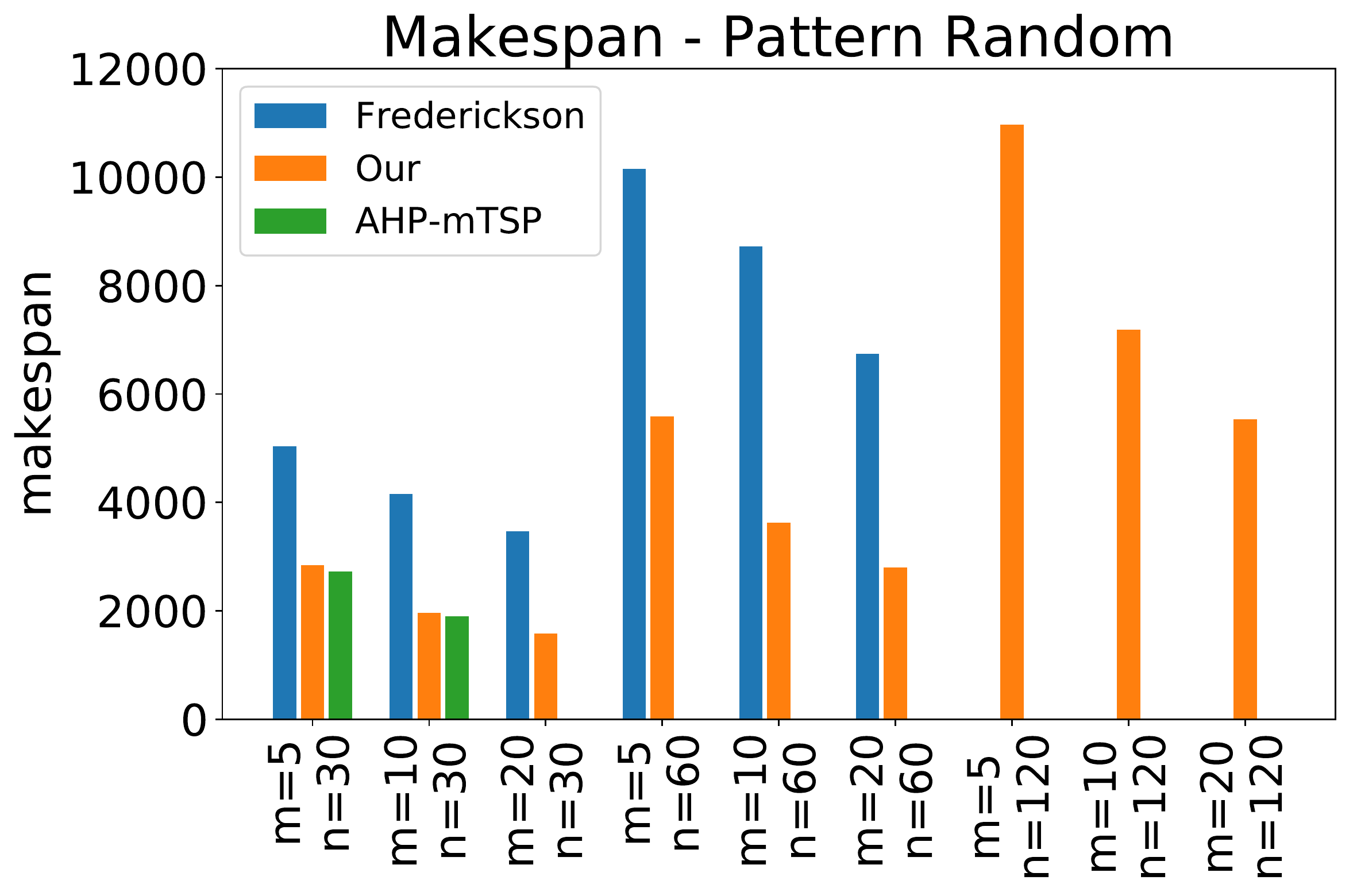}
		\caption{Makespan comparison}
		\label{fig:makespan_comparison}
	\end{subfigure}
	\hfill
	\begin{subfigure}[t]{0.3\textwidth}
		\centering
		\includegraphics[width=\textwidth]{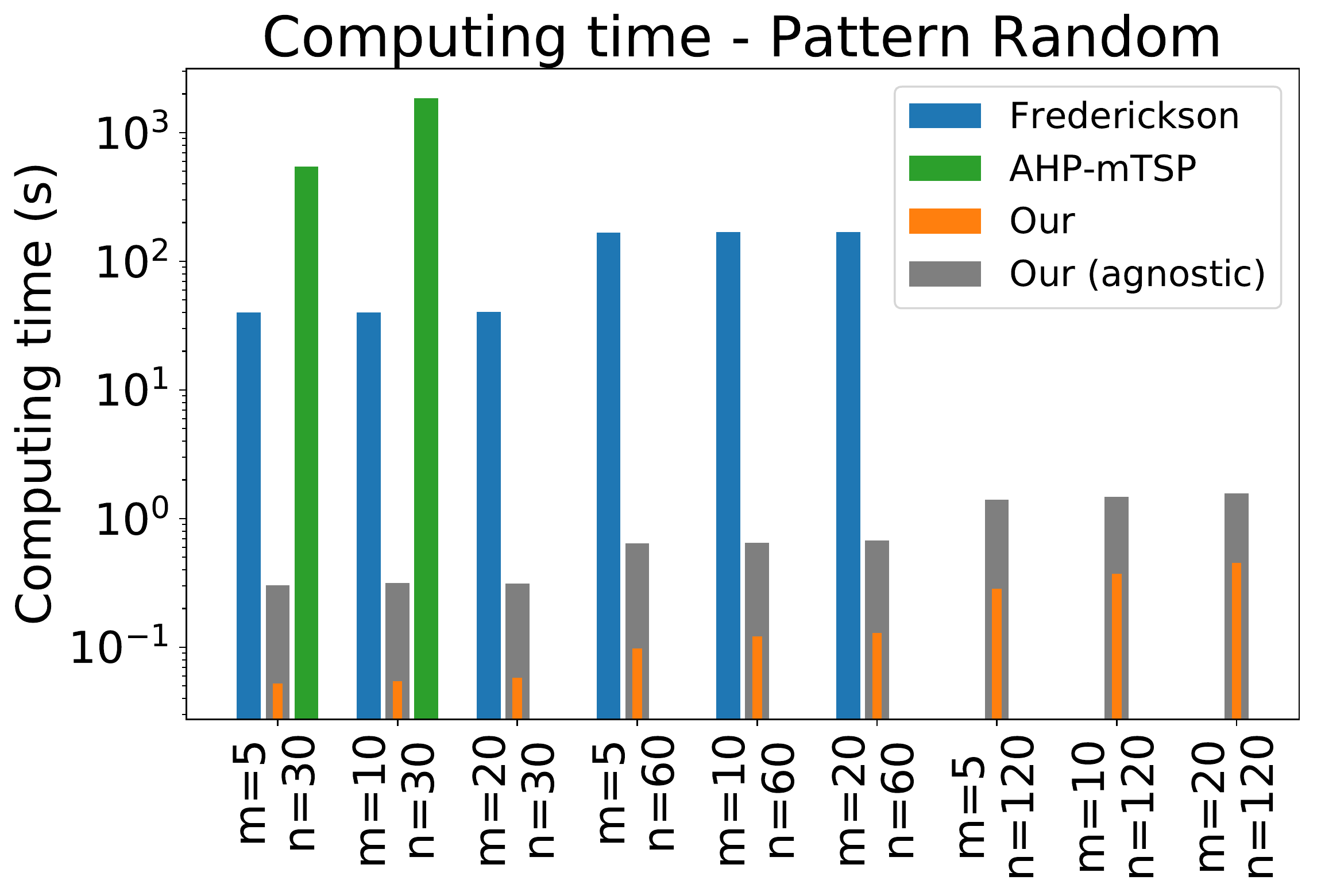}
		\caption{Computing time comparison}
		\label{fig:runtime_comparison}
	\end{subfigure}
	\hfill
	\begin{subfigure}[t]{0.3\textwidth}
		\centering
		\includegraphics[width=\textwidth]{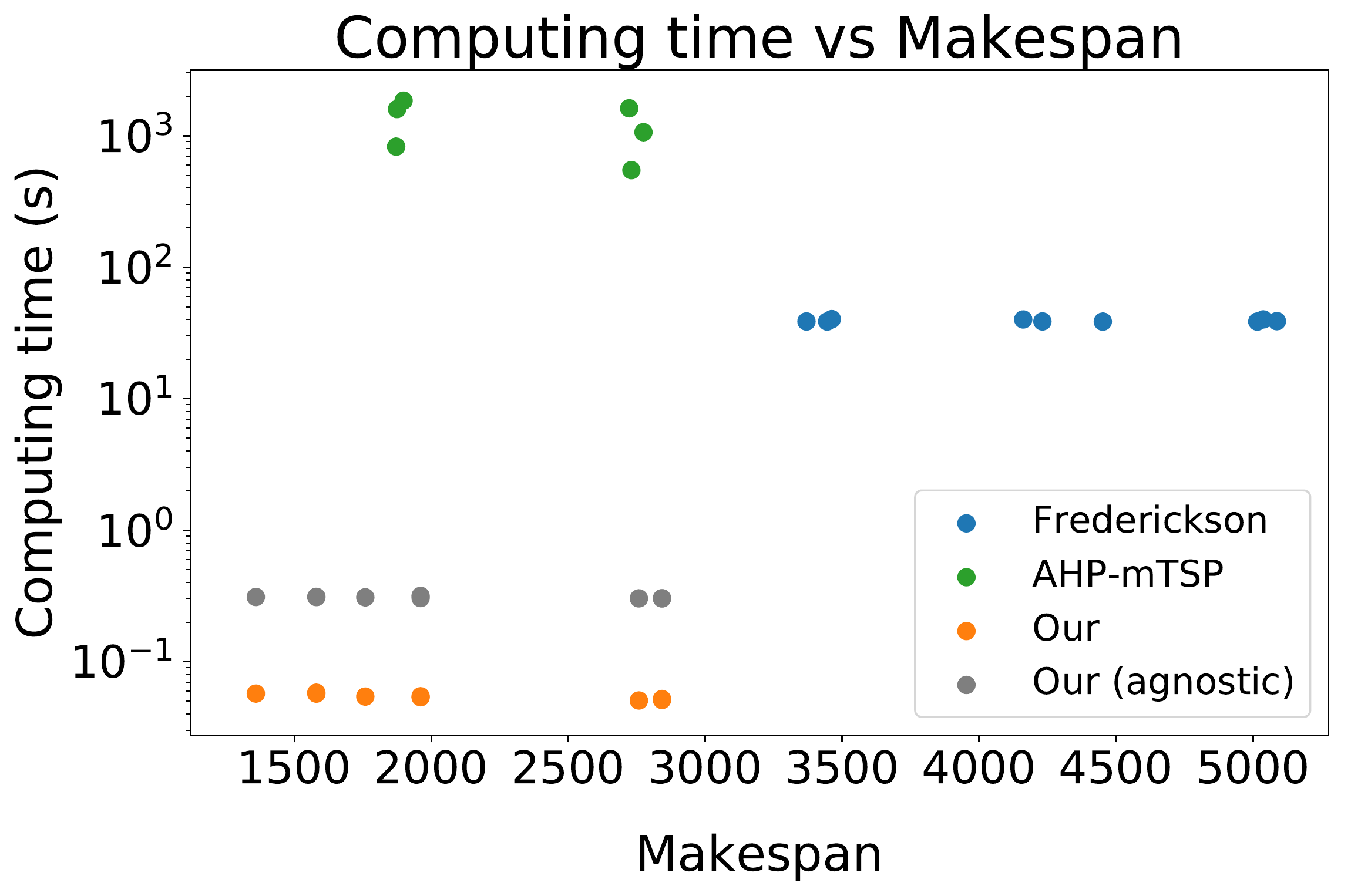}
		\caption{Computing time vs Makespan}
		\label{fig:scatter_comparison}
	\end{subfigure}

	\caption{
		Comparison of our algorithm, Frederickson, and AHP-mTSP.
		(\subref{fig:makespan_comparison}) compares the makespan of the three algorithms over instances that differ in the number of modules $n$ and the number of robots $m$.
		(\subref{fig:runtime_comparison}) shows a comparison of the computing time needed by the three algorithms.
		Note that the scale is logarithmic. These computing time values are single data points and not averages: the values are orders of magnitude different and noise does not affect the comparison.
		The agnostic version of our algorithm yields the same solutions as the base version, but it requires more time because it computes the TSP approximation of every module.
		(\subref{fig:scatter_comparison}) offers another perspective into the comparison of the three algorithms, showing all the data for the three patterns and for $m=\{5,10,20\}$, fixing $n=30$. Data points close to the origin are the best.
		Missing data points for Frederickson and AHP-mTSP are due to memory or timeout limitations.
	}
	\label{fig:comparison}
\end{figure*}
In this section, we compare our algorithm against Frederickson and AHP-mTSP in different modular environments varying the number $n$ of modules, the number $m$ of robots, and the patterns in which the three base modules \littleModule, \middleModule, and \bigModule\ are organized.
We consider three different patterns:
\begin{itemize} 
\item in environments of type \pattRandom\ modules are chosen randomly with a uniform probability (so that each base module is selected approximately $1/3$ of times with large $n$),
\item in environments of type \pattDec\, the first $n/3$ modules are of type \bigModule, then there are $n/3$ modules of type \middleModule, and the rest of the modules are of type \littleModule,
\item environments of type \pattInc\ have the same structure but reversed, starting from modules of type \littleModule\ and ending with modules of type \bigModule.
\end{itemize}
\noindent In \pattDec\ and \pattInc\ environments, what is decreasing and increasing is the size of the modules as the indexes of the modules grow.
We created a total of $27$ different environments generated by combining the three patterns, three values for the number of modules, $n=\{30, 60, 120\}$, and three values for the number of robots, $m=\{5,10,20\}$.
We fixed the distance between doorways to $20$ meters for all the environments, both because the impact of this distance has already been analyzed in \Cref{subsec:experimental_single_repeated} and because we have not observed anything relevant in changing it.

We set a timeout of $1$ hour for solving each instance. All instances have been solved by our algorithm.
Frederickson exceeds the available RAM for all the environments with $n=120$ modules. Finally, AHP-mTSP exceeds the timeout for instances with $n=60$ and $n=120$ and for instances with $n=30$ and $m=20$.

\Cref{fig:makespan_comparison} and \Cref{fig:runtime_comparison} show the experimental results for the $9$ environments of type \pattRandom. Results for the environments of types \pattDec\ and \pattInc\ are similar, also quantitatively (for given $n$ and $m$, the variations in the value of the makespan in the order of $10\%$) and are not reported here. 
Our algorithm consistently computes solutions with a makespan that is half of the makespan of the solutions returned by Frederickson. AHP-mTSP, for the instances in which it terminates within the timeout, provides solutions with comparable makespan wrt to our algorithm. However, when looking at the computing times, AHP-mTSP is orders of magnitude slower than our algorithm.

To show that our algorithm exploits the particular structure of modular environments and calculates integer solutions in a short computing time, we present the computing times for the base version, which takes advantage of knowing that the three base modules are repeated and computes the TSP approximation for the three base modules just once, and for the agnostic version, which blindly computes the TSP approximation for every module in the environment.
The agnostic version of the algorithm is slightly slower than the base version, as expected, but still much faster than Frederickson and AHP-mTSP (\Cref{fig:runtime_comparison}).
In \Cref{fig:scatter_comparison} we compare the overall performance of the three algorithms, in terms of the goodness of the solution found by the algorithms and the time required for the computation.
Our algorithm, in both the base and the agnostic version, is the only one that generates solutions with both low makespan and low computing time (see the data points close to the origin).


\section{Conclusions and future work}


In this paper, we have studied the multi-Traveling Salesperson Problem (mTSP) in modular environments, providing an efficient approximation algorithm that partitions the modules and assigns each group of adjacent modules to a different robot, obtaining a solution whose cost is within a fixed bound from the cost of an optimal solution. Experiments show that our approach effectively solves mTSP instances in large modular environments, outperforming state-of-the-art algorithms designed for mTSPs in generic environments. 

Future work will extend the results of this paper to modular environments that are not ``linear'', for example, those in which modules are arranged in trees, circles, or grids. Also, in environments with identical modules, the problem input can be exponentially compressed, implying that any algorithm running in $\textit{poly}(n)$ (included ours) is actually pseudo-polynomial in the problem instance size. This problem variant would deserve further investigation.
Another extension will consider the presence of  multiple ``doorways'' for each module, like in the case of a multi-floor building in which floors are connected to each other by different staircases and elevators. Moreover, the study of approximation algorithms that go beyond the idea of integer solutions could be addressed. Finally, steps towards the practical implementation of the proposed algorithm in robot platforms employed in real-world applications will be undertaken.

\bibliographystyle{abbrv}
\bibliography{ms}

\end{document}